\documentclass[12pt]{article}
\usepackage{amsmath, amssymb, amsthm}
\usepackage{cite}
\usepackage{arxiv}
\usepackage{subfig}
\usepackage{multirow}
\usepackage{graphicx}
\usepackage{listings}
\usepackage{xcolor}
\usepackage{url}

\title{\textbf{Explicit Factorization of $x^{p+1}-1$ over $\mathbb{Z}_{p^e}$: \\ A Structural Approach via Dickson Polynomials}}

\author{
    Yongchao Wang$^{1}$, Yang Ding$^{2}$, Jiansheng Yang$^{2}$, and Zhiqiu Huang$^{1}$ \\[0.5em]
    $^{1}$\textit{College of Computer Science and Technology, Nanjing University of Aeronautics and Astronautics, Nanjing, China} \\
    \texttt{\{wangyc, zqhuang\}@nuaa.edu.cn} \\[0.5em]
    $^{2}$\textit{Department of Mathematics, Shanghai University, Shanghai 200444, China} \\
    \texttt{dingyang@shu.edu.cn, yjsyjs@staff.shu.edu.cn}
}

\date{\today}

\theoremstyle{plain}
\newtheorem{thm}{Theorem}[section]

\newtheorem{lem}{Lemma}
\newtheorem{prop}{Proposition}

\newtheorem{defn}{Definition}
\newtheorem{rem}{Remark}

\numberwithin{equation}{section}

\begin{document}

\maketitle

\begin{abstract}
Let $p$ be an odd prime. The factorization of the polynomial $x^{p+1}-1$ over the integer residue ring $\mathbb{Z}_{p^e}$ is pivotal for constructing cyclic codes with Hermitian symmetry, a critical resource for Linear Complementary Dual (LCD) codes and Entanglement-Assisted Quantum Error-Correcting Codes (EAQECC). 
Traditionally, lifting factorizations relies on the generic Hensel's Lemma, masking the underlying algebraic structure. 
In this paper, we establish a structural isomorphism between the lifting process and the roots of a special auxiliary polynomial $V(x)$, unveiling a deterministic link to Dickson polynomials. 
Based on this theory, we develop \texttt{Dickson-Engine}, a linear-time algorithm ($O(ep)$) that outperforms standard libraries by orders of magnitude.
Applying this engine to $\mathbb{Z}_{169}$, we explicitly construct a family of classical LCD codes of length $n=182$ via the isometric Gray map. 
Our search reveals codes with parameters (e.g., $[182, 1, 168]_{13}$ and $[182, 2, 144]_{13}$) that are \textbf{near-optimal} with respect to the theoretical Griesmer Bound. 
Notably, we discover a ``robustness plateau'' starting from non-trivial dimensions ($k=4$), where the minimum distance remains stable ($d=120$) even as the dimension triples ($k=4 \rightarrow 12$). These codes provide exceptional resources for post-quantum cryptography and quantum error correction without entanglement consumption ($c=0$).
\end{abstract}

\bigskip

\section{Introduction}

The factorization of polynomials over finite rings, particularly over the integer residue ring $\mathbb{Z}_{p^e}$, has long been a foundational problem in algebraic coding theory. While the case where $\gcd(n, p) = 1$ ensures unique factorization and has been extensively studied, the specific instance of $n = p+1$ has recently garnered significant attention, transcending its classical roots to become a cornerstone in the quantum era.

\subsection*{The Renaissance of $x^{p+1}-1$: From Classical to Quantum}
In the landscape of modern coding theory, $x^{p+1}-1$ represents a unique ``stability island.'' Unlike the complex repeated-root codes~\cite{Dinh2018}, $x^{p+1}-1$ belongs to the simple-root category (since $\gcd(p+1, p)=1$). Its roots in the Galois Ring $GR(p^e, 2)$ naturally satisfy the conjugacy relation $\zeta^p = \zeta^{-1}$, inducing a rigid \textbf{Hermitian inner product structure}.
This property is critical for constructing \textbf{Entanglement-Assisted Quantum Error-Correcting Codes (EAQECC)}~\cite{Sidana2022} and \textbf{Linear Complementary Dual (LCD)} codes, which are essential for lattice-based Post-Quantum Cryptography to resist hull attacks~\cite{Nishimura2025}.

\subsection*{The Challenge: Moving Beyond Iteration}
Despite its rising importance, obtaining the explicit irreducible factors of $x^{p+1}-1$ over $\mathbb{Z}_{p^e}$ remains a non-trivial computational task. 
The standard approach relies on \textbf{Hensel's Lemma}, which lifts factors from $\mathbb{F}_p$ to $\mathbb{Z}_{p^e}$. While theoretically guaranteeing existence and uniqueness, standard Hensel lifting is inherently \textbf{iterative}, treating the coefficients as black-box values to be refined step-by-step.
A pivotal attempt to structuralize this process was made by McGuire~\cite{McGuire2001}, who proposed lifting \textbf{Newton sums} instead of coefficients. While this approach cleverly exploits the linearity of the trace map, it remains bound to the iterative nature of lifting ($k \to 2k \to \dots$) and incurs additional computational overhead when converting sums back to coefficients via Newton-Girard formulas. Similarly, powerful lattice-based methods like those by van Hoeij~\cite{Hoeij2002}, while dominant over $\mathbb{Q}$, are often computationally excessive for the local geometry of $\mathbb{Z}_{p^e}$.

\textbf{The fundamental question remains:} Can we bypass these iterative approximations? Is there a hidden algebraic pattern that determines the lifted coefficients \textit{a priori}?

\subsection*{Our Contributions}

In this paper, we bridge the gap between the algebraic theory of Dickson polynomials and the explicit construction of coding resources for the quantum era. Our main contributions are threefold:

\begin{itemize}
    \item \textbf{A Structural Isomorphism via Dickson Polynomials:} 
    We reveal that the factorization of $x^{p+1}-1$ over $\mathbb{Z}_{p^e}$ is not arbitrary but strictly governed by the recursive structure of Dickson polynomials. By introducing a novel auxiliary polynomial $V(x)$, we transform the complex Hensel lifting process into a structural isomorphism that maps the problem to a scalar root-finding task. This theoretical breakthrough unifies the factorization logic for all odd primes.

    \item \textbf{Linear-Time Algorithm and Open-Source Engine:} 
    Based on the $V(x)$ mechanism, we propose a deterministic algorithm with linear complexity $O(e \cdot p)$, bypassing the heavy polynomial arithmetic of standard methods (e.g., NTL's $O(p^2)$). We release \textbf{\texttt{Dickson-Engine}}\footnote{\url{https://github.com/Nothing256/Dickson-Engine}}, a high-performance C implementation demonstrating $>$300$\times$ speedup.

    \item \textbf{Construction of Near-Optimal LCD Codes:} 
    Leveraging the computational power of our engine, we performed an extensive search for Linear Complementary Dual (LCD) codes over $\mathbb{Z}_{13^2}$ with length $n=182$. We successfully constructed a family of LCD codes that are \textbf{near-optimal} with respect to the theoretical \textbf{Griesmer Bound} (e.g., $[182, 1, 168]_{13}$ and $[182, 2, 144]_{13}$). Furthermore, we identified a ``robustness plateau'' where the minimum distance remains stable ($d=120$) even as the dimension triples ($k=4 \to 12$). These codes provide exceptional resources for post-quantum cryptography and entanglement-assisted quantum error correction without entanglement consumption ($c=0$).
\end{itemize}

This work transforms the factorization problem from a routine computation into a study of the intrinsic symmetry of the ring $\mathbb{Z}_{p^e}$, providing a transparent algebraic tool that aligns with the explicit construction needs of modern quantum and lattice-based cryptography.

\section{Preliminaries and Algebraic Tools}
\label{sec:preliminaries}

In this section, we establish the necessary notations and review key algebraic tools that underpin our factorization framework. Throughout this paper, let $\mathbb{Z}$ denote the ring of integers. For a prime power $q=p^e$, where $p$ is an odd prime and $e \ge 1$, $\mathbb{F}_q$ denotes the finite field with $q$ elements, and $\mathbb{Z}_{p^e}$ denotes the residue class ring $\mathbb{Z}/p^e\mathbb{Z}$. 

Our primary focus is the factorization of $x^n - 1$ over the local ring $\mathbb{Z}_{p^e}$, specifically for the case $n=p+1$.

\subsection{Cyclotomic Cosets and Factorization over $\mathbb{F}_p$}

The factorization of $x^n-1$ over the base field $\mathbb{F}_p$ is governed by the partition of indices modulo $n$.
Let $n$ be a positive integer coprime to $p$. The \textit{$p$-cyclotomic coset} modulo $n$ containing $i$ is defined as:
$$C_i = \{i \cdot p^j \pmod n \mid j=0, 1, \dots, r-1\},$$
where $r$ is the smallest positive integer such that $i \cdot p^r \equiv i \pmod n$. A subset $S \subseteq \{0, 1, \dots, n-1\}$ is called a \textit{complete set of representatives} if the cosets $\{C_i\}_{i \in S}$ form a partition of $\{0, 1, \dots, n-1\}$.

The following proposition is a standard result linking these cosets to irreducible factors:

\begin{prop}[\cite{Lidl1997,Huffman2003}]\label{prop:cyc_coset}
Let $n$ be coprime to $p$, and let $\alpha$ be a primitive $n$-th root of unity in an extension of $\mathbb{F}_p$.
\begin{itemize}
    \item[(i)] For any $s \in \{0, \dots, n-1\}$, the minimal polynomial of $\alpha^s$ over $\mathbb{F}_p$ is given by
    $$M_{\alpha^s}(x) = \prod_{i \in C_s} (x - \alpha^i),$$
    where $C_s$ is the $p$-cyclotomic coset containing $s$.
    \item[(ii)] The canonical factorization of $x^n - 1$ over $\mathbb{F}_p$ is
    $$x^n - 1 = \prod_{s \in S} M_{\alpha^s}(x),$$
    where $S$ is a complete set of representatives.
\end{itemize}
\end{prop}

\subsection{Hensel's Lifting: The Bridge to $\mathbb{Z}_{p^e}$}

To lift the factorization from $\mathbb{F}_p$ to the ring $\mathbb{Z}_{p^e}$, the standard tool is Hensel's Lemma. Since we are dealing with a local ring structure, we recall the specific form of the lifting property.

\begin{prop}[\cite{Lang1994}]\label{prop:hensel_lemma}
Let $f(x) \in \mathbb{Z}[x]$ be a polynomial and $a$ be an integer. Assume that $f(a) \equiv 0 \pmod{p^e}$ for some integer $e \geq 1$. If $f'(a) \not\equiv 0 \pmod p$, then there exists a unique integer $u_a \in \{0, \dots, p-1\}$ such that
$$f(a + u_a p^e) \equiv 0 \pmod{p^{e+1}}.$$
Explicitly, the lifting term $u_a$ is determined by:
\begin{equation}\label{eq:hensel_update}
u_a \equiv - \frac{f(a)/p^e}{f'(a)} \pmod p.
\end{equation}
\end{prop}

\noindent \textbf{Remark.} While Proposition~\ref{prop:hensel_lemma} guarantees the existence and uniqueness of the lift, applying it directly to polynomial factorization requires iterative computations for each coefficient. One of the main contributions of this paper is to bypass this coefficient-wise iteration by encapsulating the lifting logic into a single auxiliary polynomial $V(x)$.

\subsection{Dickson Polynomials: The Hidden Generator}

A key insight of our work involves the unexpected role of Dickson polynomials. Originally introduced in the context of permutation polynomials, and later analyzed for their factorization properties~\cite{Chou1997,Fitzgerald2005}, Dickson polynomials relate sums of powers to elementary symmetric polynomials.

\begin{defn}[\cite{Lidl1997}]\label{def:dickson}
Let $R$ be a commutative ring with identity. For $\gamma \in R$, the \textit{Dickson polynomial of the first kind} of degree $n$, denoted $D_n(x, \gamma)$, is defined as:
\begin{equation}\label{eq:dickson_formula}
D_n(x, \gamma) = \sum_{i=0}^{\lfloor n/2 \rfloor} \frac{n}{n-i} \binom{n-i}{i} (-\gamma)^i x^{n-2i}.
\end{equation}
These polynomials satisfy the fundamental recurrence relation:
$$D_n(x, \gamma) = x D_{n-1}(x, \gamma) - \gamma D_{n-2}(x, \gamma), \quad \text{for } n \geq 2,$$
with initial conditions $D_0(x, \gamma) = 2$ and $D_1(x, \gamma) = x$.
\end{defn}

The structural utility of Dickson polynomials in our context stems from their connection to power sums of roots, often referred to as Waring's Formula:

\begin{prop}[\cite{Lidl1997}]\label{prop:dickson_waring}
Let $x_1, x_2$ be indeterminates over a ring $R$. Then:
$$x_1^n + x_2^n = D_n(x_1 + x_2, x_1 x_2).$$
\end{prop}

\noindent \textbf{Observation.} While typically used to study field permutations, we will demonstrate in Section \ref{sec:framework} that $D_{p-1}(x, 1)$ acts as a \textbf{coefficient generator} for the factorization of $x^{p+1}-1$, providing a direct algebraic link between the field characteristic $p$ and the lifted coefficients in $\mathbb{Z}_{p^e}$.

\subsection{Homogeneous Weight and Gray Map}
\label{sec:gray_map}

To evaluate the error-correcting performance in the ambient space, we employ the generalized Gray map, which links codes over rings to codes over finite fields. 
For the ring $\mathbb{Z}_{p^e}$, we utilize the \textbf{homogeneous weight} $w_{\text{hom}}$. 
In the specific case of $e=2$, which corresponds to our computational experiments, the homogeneous weight of an element $u \in \mathbb{Z}_{p^2}$ is explicitly defined as:
\begin{equation}
    w_{\text{hom}}(u) = 
    \begin{cases} 
    0 & \text{if } u = 0, \\
    p & \text{if } u \in p\mathbb{Z}_{p^2} \setminus \{0\} \quad (\text{Zero Divisors}), \\
    p-1 & \text{if } u \in \mathbb{Z}_{p^2}^\times \quad (\text{Units}).
    \end{cases}
\end{equation}
This weight function serves as the algebraic counterpart to the Hamming weight in the image space under the Gray map. 
As introduced by Constantinescu and Heise~\cite{constantinescu1997metric}, this metric is naturally adaptable to the structure of chain rings.

Note that we adopt the un-normalized form of the homogeneous weight here. This choice is deliberate to ensure that $w_{\text{hom}}(u)$ equals the Hamming weight of its Gray image $\Phi(u)$ in the ambient space $\mathbb{F}_p^p$, thereby allowing us to measure the code's minimum distance as an integer value.

For the specific case of $e=2$ (i.e., $\mathbb{Z}_{p^2}$), any element $u \in \mathbb{Z}_{p^2}$ can be uniquely represented as $u = a + bp$, where $a, b \in \mathbb{F}_p$. The Gray map $\Phi: \mathbb{Z}_{p^2} \to \mathbb{F}_p^p$ is defined as~\cite{greferath1999gray}:
\begin{equation}
    \Phi(a + bp) = (b, b+a, b+2a, \dots, b+(p-1)a).
\end{equation}
This map extends coordinate-wise to $\mathbb{Z}_{p^2}^n \to \mathbb{F}_p^{np}$. 
A crucial property of $\Phi$ is that it is a distance-preserving isometry from $(\mathbb{Z}_{p^2}^n, w_{\text{hom}})$ to $(\mathbb{F}_p^{np}, w_{\text{Hamming}})$~\cite{greferath1999gray}. Consequently, a cyclic code $\mathcal{C}$ of length $n$ over $\mathbb{Z}_{p^2}$ is mapped to a linear code of length $N = np$ over $\mathbb{F}_p$, preserving the minimum distance structure.

\section{The Explicit Factorization Framework}
\label{sec:framework}

In this section, we present our main structural results. We decompose the factorization problem into two stages: first, establishing a generative mechanism for the coefficients over the base field $\mathbb{F}_p$; and second, introducing the novel $V(x)$ polynomial to explicitly lift these coefficients to the ring $\mathbb{Z}_{p^e}$.

\subsection{Base Layer: Generating Coefficients via Dickson Mapping}

Let $p$ be an odd prime. The factorization structure depends on the specific congruence class of $p$ modulo 4. While the majority of factors are quadratic, there are specific linear and quadratic factors corresponding to the fixed points of the Frobenius automorphism.

\begin{thm}\label{thm:base_factorization}
The irreducible factorization of $x^{p+1}-1$ over $\mathbb{F}_p$ is given by:
\begin{equation}
x^{p+1}-1 = (x-1)(x+1) \Psi(x) \prod_{i \in S'} (x^2 - A_i x + 1),
\end{equation}
where:
\begin{itemize}
    \item The linear factors $(x-1)$ and $(x+1)$ correspond to the cosets $\{0\}$ and $\{(p+1)/2\}$.
    \item The term $\Psi(x)$ handles the case of $x^2+1$:
    $$ \Psi(x) = \begin{cases} x^2+1 & \text{if } p \equiv 3 \pmod 4 \\ \text{1 (empty product)} & \text{if } p \equiv 1 \pmod 4 \end{cases} $$
    \item $S'$ is the set of representatives for the remaining cosets of size 2, defined as $S' = \{1, 2, \dots, \frac{p+1}{2}\} \setminus \{\frac{p+1}{4}\}$ if $p \equiv 3 \pmod 4$, and $S' = \{1, 2, \dots, \frac{p+1}{2}\}$ if $p \equiv 1 \pmod 4$. Note that the index $\frac{p+1}{4}$ corresponds to the factor $x^2+1$ (where $A_i=0$). The coefficients $A_i$ for these factors are explicitly generated as follows:
    \begin{enumerate}
        \item \textbf{Seed Generation:} Let $f(x) = x^2 - a_1 x + a_2 \in \mathbb{F}_p[x]$ be a primitive polynomial (where $a_1$ and $a_2$ correspond to the trace and norm of the primitive element $\beta \in \mathbb{F}_{p^2}$, respectively). The generator $A_1$ is obtained via:
        $$A_1 = D_{p-1}(a_1, a_2).$$
        \item \textbf{Recursive Expansion:} For $i \ge 2$, coefficients are generated by the recurrence:
        $$A_i = D_i(A_1, 1) = A_1 A_{i-1} - A_{i-2}.$$
    \end{enumerate}
\end{itemize}
\end{thm}

\begin{proof}
The factorization structure in (i) and the specific linear/quadratic forms follow directly from the size of the $p$-cyclotomic cosets modulo $p+1$, which is at most 2 since $p^2 \equiv 1 \pmod{p+1}$.

We focus on proving the generative mechanism for $A_i$. Let $\beta$ be a root of the primitive polynomial $f(x) = x^2 - a_1 x + a_2$ in $\mathbb{F}_{p^2}$. Since $f(x)$ is primitive, $\beta$ is a generator of $\mathbb{F}_{p^2}^*$. We set $\alpha = \beta^{p-1}$. Since the order of $\beta$ is $p^2-1$, the order of $\alpha$ is $(p^2-1)/(p-1) = p+1$, making $\alpha$ a primitive $(p+1)$-th root of unity.

The coefficient for the $i$-th factor is $A_i = \alpha^i + \alpha^{-i}$. Recall the fundamental property of Dickson polynomials (Waring's formula): $u^k + v^k = D_k(u+v, uv)$.
\begin{itemize}
    \item \textbf{For the seed $A_1$:} We have $A_1 = \alpha + \alpha^{-1} = \beta^{p-1} + (\beta^{-1})^{p-1}$. Note that in $\mathbb{F}_{p^2}$, the conjugate of $\beta$ is $\beta^p$. Thus, the term $\alpha^{-1}$ corresponds to the conjugate action on the primitive root relative to the norm. Specifically, we observe:
    $$A_1 = \beta^{p-1} + (\beta^p)^{p-1} = D_{p-1}(\beta+\beta^p, \beta^{p+1}).$$
    From Vieta's formulas on $f(x)$, we have $\beta + \beta^p = a_1$ and $\beta^{p+1} = a_2$. Therefore, $A_1 = D_{p-1}(a_1, a_2)$.
    
    \item \textbf{For the recursion:} Since $A_i = \alpha^i + (\alpha^{-1})^i$ and $\alpha \cdot \alpha^{-1} = 1$, we simply apply the Dickson property again with $u=\alpha, v=\alpha^{-1}$:
    $$A_i = D_i(\alpha + \alpha^{-1}, \alpha \alpha^{-1}) = D_i(A_1, 1).$$
    This satisfies the recurrence $D_i(x, 1) = x D_{i-1}(x, 1) - D_{i-2}(x, 1)$.
\end{itemize}
This confirms that the entire set of coefficients is deterministically generated by the Dickson mapping of the primitive polynomial's coefficients.
\end{proof}

This theorem transforms the factorization from a search problem into a \textit{generation} problem, providing a compact representation of the factors.

\subsection{Structural Lifting: The $V(x)$ Mechanism}

The classical approach to factor over $\mathbb{Z}_{p^e}$ involves lifting each quadratic factor $x^2 - A_i x + 1$ individually. However, this overlooks a crucial symmetry in the solution space derived from the properties of primitive roots.

Consider the relationship between the coefficient $A_i$ and the coefficient for the symmetric coset index $j = \frac{p+1}{2} - i$. Since $\alpha^{\frac{p+1}{2}} = -1$, we can explicitly derive:
\begin{equation}
    A_j = \alpha^j + \alpha^{-j} = \alpha^{\frac{p+1}{2}-i} + \alpha^{-(\frac{p+1}{2}-i)} = (-1)\alpha^{-i} + (-1)\alpha^i = -A_i.
\end{equation}
This derivation confirms that the factors $x^2 - A_i x + 1$ and $x^2 + A_i x + 1$ generally appear in pairs. Consequently, the algebraic constraints on $A_i$ are most naturally expressed through the invariant of this symmetry pairing:
\begin{equation}
    (x^2 - A_i x + 1)(x^2 + A_i x + 1) = x^4 + (2 - A_i^2)x^2 + 1.
\end{equation}
This motivates us to define the \textbf{structural variable} $S_i = 2 - A_i^2$. By shifting our focus from lifting $A_i$ to lifting $S_i$, we transform the problem into a domain where the constraints become linear and recursive.

\subsubsection*{Heuristic Observation: The Coefficient Triangle}

Our discovery of $V(x)$ was not derived from abstract principles \textit{ab initio}, but from an empirical investigation of the lifting constraints.
We observed that for any prime $p$, the structural parameters $S_i$ satisfy a unifying polynomial equation $V(S) \equiv 0 \pmod p$.
While for small primes like $p=5$, the polynomial is trivial ($S-1$), calculating for larger primes revealed a striking pattern. The coefficients of these polynomials were not random integers; they mirrored a specific diagonal slice of Pascal's Triangle (binomial coefficients).

Let us illustrate this evolution with $p=13$, $p=17$, and $p=29$. Note how the coefficients (top line) perfectly align with the alternating binomial forms (bottom line):

\begin{itemize}
    \item \textbf{Case $p=13$} ($k=3$, degree 3):
    $$
    \begin{aligned}
    V(x) &= 1 \cdot x^3 - 1 \cdot x^2 - 2 \cdot x + 1 \\
         &= \binom{3}{0} x^3 - \binom{2}{0} x^2 - \binom{2}{1} x + \binom{1}{1}
    \end{aligned}
    $$

    \item \textbf{Case $p=17$} ($k=4$, degree 4):
    $$
    \begin{aligned}
    V(x) &= 1 \cdot x^4 - 1 \cdot x^3 - 3 \cdot x^2 + 2 \cdot x + 1 \\
         &= \binom{4}{0} x^4 - \binom{3}{0} x^3 - \binom{3}{1} x^2 + \binom{2}{1} x + \binom{2}{2}
    \end{aligned}
    $$

    \item \textbf{Case $p=29$} ($k=7$, degree 7):
    The polynomial $V(x) = x^7 - x^6 - 6x^5 + 5x^4 + 10x^3 - 6x^2 - 4x + 1$ maps perfectly to:
    $$
    \dots - \binom{6}{1} x^5 + \binom{5}{1} x^4 + \binom{5}{2} x^3 - \binom{4}{2} x^2 - \binom{4}{3} x + \binom{3}{3}
    $$
\end{itemize}

This visual alignment provided the decisive intuition: the coefficients are generated by terms of the form $(-1)^s \binom{n-s}{m}$. This ``coefficient triangle'' led us to formulate the general Definition~\ref{def:Vx}, transforming a computational observation into a rigorous algebraic identity.
\subsubsection*{The Auxiliary Polynomial $V(x)$}

Based on this observation, we introduce the polynomial $V(x) \in \mathbb{Z}[x]$, which serves as the core engine of our lifting algorithm.

\begin{defn}[The Structural Polynomial $V(x)$]\label{def:Vx}
Let $k = \lfloor \frac{p}{4} \rfloor$. We define $V(x) = \sum_{r=0}^{k} V_r x^{k-r}$ with coefficients determined by the prime $p$:
\begin{itemize}
    \item If $p = 4k + 3$:
    $$V_{2i} = (-1)^i \binom{k-i}{i}, \quad V_{2i+1} = 0.$$
    \item If $p = 4k + 1$:
    $$V_{2i} = (-1)^i \binom{k-i}{i}, \quad V_{2i+1} = (-1)^{i+1} \binom{k-i-1}{i}.$$
\end{itemize}
The construction of these coefficients relies on combinatorial identities discussed in \cite{Shi2009}.
\end{defn}

The significance of $V(x)$ lies in the following isomorphism theorem, which connects the algebraic lifting process to the roots of $V(x)$.

\begin{thm}[Structural Isomorphism]\label{thm:isomorphism}
Recall that $S_i = 2 - A_i^2$ is our structural variable. 
\emph{Remarkably, by the Dickson property $A_{2i} = A_i^2 - 2$, this transformation implies an intrinsic connection: $S_i = -A_{2i}$.}
This mapping establishes a structural isomorphism between the factorization lifting problem and the root-finding problem: the condition for lifting $S_i^{(h)}$ to step $h+1$ is isomorphic to finding roots of $V(x)$ in $\mathbb{Z}_{p^h}$:
\begin{equation}
    V(S_i^{(h)}) \equiv 0 \pmod{p^h}.
\end{equation}
The lifted value is $S_i^{(h+1)} = S_i^{(h)} + u_i V(S_i^{(h)})$, where $u_i \equiv -[V'(S_i^{(1)})]^{-1} \pmod p$.
\end{thm}

\begin{proof}
The theorem rests on establishing that $V(x)$, as defined via specific binomial coefficients, satisfies the recurrence relations derived from the elementary symmetric polynomials of the roots $S_i$. This connection essentially transforms the Hensel lifting condition into a root-finding problem for $V(x)$.
For the sake of flow and readability, we defer the rigorous algebraic derivation involving Vieta's formulas and the specific combinatorial identity to \textbf{Appendix A}.
\end{proof}

This theorem is the heart of our contribution. It implies that \textbf{the factorization of $x^{p+1}-1$ is explicitly encoded in the roots of the fixed polynomial $V(x)$}. We do not need to lift polynomials; we only need to lift the roots of $V(x)$.

\subsection{The Explicit Algorithm}

Combining the base layer generation and the $V(x)$ mechanism, we propose the following deterministic algorithm.

\noindent\textbf{Algorithm 1: Explicit Factorization via $V(x)$}
\begin{enumerate}
    \item \textbf{Seed Discovery (over $\mathbb{F}_p$):}
    \begin{itemize}
        \item Select a primitive polynomial $f(x)$ and compute the seed $A_1 = D_{p-1}(a_1, a_2)$.
        \item Compute the initial structural variable $S_1^{(1)} = 2 - (A_1^{(1)})^2 \pmod p$.
    \end{itemize}
    
    \item \textbf{Pre-computation:}
    \begin{itemize}
        \item Construct $V(x)$ according to Definition~\ref{def:Vx}.
        \item Compute the invariant update factor $C \equiv -[V'(S_1^{(1)})]^{-1} \pmod p$.
    \end{itemize}
    
    \item \textbf{Lifting Loop (for $h = 1$ to $e-1$):}
    \begin{itemize}
        \item Lift \textbf{only the single seed} $S_1$ via $V(x)$-Newton iteration:
        \item Compute the lifting error: $\Delta = V(S_1^{(h)}) / p^h$.
        \item Update the structural parameter: $S_1^{(h+1)} = S_1^{(h)} + (C \cdot \Delta \pmod p) \cdot p^h$.
    \end{itemize}
    
    \item \textbf{Expansion at Target Precision:}
    \begin{itemize}
        \item Recover $A_1^{(e)}$ from $S_1^{(e)}$ (see Remark below).
        \item Generate all remaining coefficients via the Dickson recurrence: $A_i^{(e)} = D_i(A_1^{(e)}, 1)$.
    \end{itemize}
    
    \item \textbf{Output:} The factors $x^2 - A_i^{(e)} x + 1$ over $\mathbb{Z}_{p^e}$.
\end{enumerate}

This algorithm lifts only a \textbf{single scalar value} $S_1$ through all $e$ layers, deferring the full coefficient expansion to the final precision. This avoids the complexity of polynomial modular arithmetic entirely, reducing the problem to simple integer arithmetic.

\begin{rem}[Efficient Recovery of $A_1^{(e)}$]
The final step recovers the factorization coefficient $A_1^{(e)}$ from the structural variable $S_1^{(e)}$. Since $S_1 = 2 - A_1^2$, solving for $A_1$ is equivalent to finding a square root of $2 - S_1$ in $\mathbb{Z}_{p^e}$. Instead of employing generic modular square root algorithms (e.g., lifting Tonelli-Shanks solutions via Hensel's Lemma~\cite{Cohen1993}), which can be computationally intensive, we leverage the \textbf{already known} base layer coefficient $A_1^{(1)} \in \mathbb{F}_p$. By treating $A_1^{(1)}$ as an initial approximation, we can apply \textbf{Newton's method} to lift the root to precision $e$ with quadratic convergence. This aligns the complexity of the recovery step with the rest of the lifting process ($O(e)$ scalar operations), ensuring the entire algorithm remains strictly linear in both $p$ and $e$.
\end{rem}

\subsection{Discussion: Efficiency, Comparisons, and Implications}
\label{subsec:discussion}

Our framework offers a distinct departure from general-purpose factorization algorithms. Here we discuss its computational characteristics, contrast it with established methods, and highlight its relevance to modern cryptographic security.

\textbf{Dependency on Primitive Polynomials.}
The initialization of our algorithm (Step 1) relies on the availability of a primitive polynomial of degree 2 over $\mathbb{F}_p$ to generate the seed $A_1$. This dependency is computationally negligible given the high density of primitive polynomials, which allows for efficient random search even for large primes~\cite{Lidl1997}.

\textbf{Comparison with General Solvers (Finite Fields).} 
It is illuminating to contrast our framework with the classical \textbf{Berlekamp's algorithm}~\cite{Berlekamp1968}, the universal standard for factoring polynomials over finite fields. Berlekamp's method relies on linear algebra to compute the kernel of the Frobenius map, typically involving computationally expensive matrix operations ($O(n^3)$) and probabilistic splitting steps. Specifically, the \textbf{Cantor-Zassenhaus (CZ) algorithm}~\cite{Cantor1981} separates factors by computing GCDs with random polynomials throughout the recursion, introducing non-determinism at every stage.

In the specific context of $x^{p+1}-1$, however, our approach exploits the inherent \textbf{Hermitian symmetry}. By identifying the factorization coefficients $A_i$ as \textbf{traces} generated by the recursive structure of \textbf{Dickson polynomials}, we replace the generic matrix reductions of Berlekamp's method with a deterministic \textbf{scalar recurrence}.
This represents a fundamental \textbf{structural shortcut}: instead of ``discovering'' the factors through probabilistic linear algebra, we ``generate'' them directly from the field's arithmetic structure.

\textbf{Comparison with Lifting Methods (Rings).}
For lifting to $\mathbb{Z}_{p^e}$, methods based on Lattice Reduction (LLL) (e.g., Van Hoeij~\cite{Hoeij2002}) are powerful but computationally excessive for local rings. Similarly, McGuire's approach via Newton sums~\cite{McGuire2001} avoids coefficient arithmetic but necessitates iterative lifting steps ($k \to 2k$) and complex Newton-Girard conversions.
Our $V(x)$ mechanism acts as a \textbf{closed-form solution} that effectively ``short-circuits'' these iterations. It transforms the problem into a simple scalar recurrence, avoiding matrices, lattices, and polynomial arithmetic entirely.

\textbf{Complexity, Efficiency, and Granularity.}
Our algorithmic framework offers a versatile trade-off between total throughput and structural flexibility, supporting two distinct operational modes:

\textit{1. Global Linear Efficiency (The "Generator" Mode):} 
For full factorization, we leverage the recursive nature of Dickson polynomials. Instead of lifting all factors individually, we lift a single generator seed $S_1$ to $A_1$ and then derive the remaining coefficients via the scalar recurrence $A_i = A_1 A_{i-1} - A_{i-2}$. This optimization dramatically reduces the total time complexity from $O(e \cdot k^2)$ to $\mathbf{O(e \cdot k)} \approx \mathbf{O(e \cdot p)}$. This \textbf{linear complexity} explains the flat performance curve observed in our experiments, where our engine outperforms standard $O(p^2)$ algorithms by orders of magnitude for large primes.

\textit{2. Structural Granularity (The "Targeted" Mode):} 
Beyond raw speed, this generation-based paradigm offers a unique architectural advantage: \textbf{independence}. Unlike monolithic algorithms that must compute the entire factorization structure globally, our mechanism allows for the \textbf{targeted lifting} of any specific factor $S_i$ without computing the others. This decoupling opens the door for massive parallelism, a flexibility that traditional iterative lifting lacks.

\textbf{Implications for Post-Quantum Security.}
The determinism and speed of our algorithm have direct implications for the security of lattice-based cryptography constructed from LCD codes. As highlighted by Nishimura et al.~\cite{Nishimura2025}, the security of such systems against \textbf{Hull Attacks} depends on the precise algebraic decomposition of the underlying polynomial. Our explicit algorithm allows for the rapid identification of factors with specific symmetric properties (i.e., self-reciprocal factors), providing a crucial tool for both cryptanalysis and secure system design.

\section{Algorithmic Implementation and Verification}
\label{sec:examples}

To demonstrate the efficiency and structural clarity of our framework, we apply Algorithm 1 to two concrete examples. These examples illustrate how the $V(x)$ mechanism bypasses polynomial modular arithmetic, reducing the lifting process to simple integer operations.

\subsection{Example 1: Factorization of $x^{14}-1$ over $\mathbb{Z}_{13^2}$}

Let $p=13$ and $n=14$. We aim to factor $x^{14}-1$ over the ring $\mathbb{Z}_{169}$.
Here $k = \lfloor 13/4 \rfloor = 3$. This implies the auxiliary polynomial $V(x)$ will be of degree 3.

\noindent \textbf{Step 1: Initialization over $\mathbb{F}_{13}$.}
We select a primitive polynomial $f(x) = x^2 + x + 2$ over $\mathbb{F}_{13}$.
Using the Dickson mapping (Theorem~\ref{thm:base_factorization}), the generator coefficient is $A_1 = D_{12}(-1, 2) \equiv 5 \pmod{13}$.
From the recurrence $A_i = A_1 A_{i-1} - A_{i-2}$, we generate the remaining coefficients:
\[
A_2 \equiv 10, \quad A_3 \equiv 6 \pmod{13}.
\]
Thus, the factorization over $\mathbb{F}_{13}$ is determined.

\noindent \textbf{Step 2: Pre-computation of $V(x)$.}
According to Definition~\ref{def:Vx} (case $p=4k+1$), the structural polynomial is:
\[
V(x) = x^3 - x^2 - 2x + 1.
\]
Following our single-seed mechanism, we only need to extract the scalar seed for the main generator $A_1 \equiv 5$:
\[
S_1^{(1)} \equiv 2 - (A_1)^2 \equiv 2 - 25 \equiv 3 \pmod{13}.
\]
(Note: Since the subsequent process relies on the Dickson recurrence, there is no need to compute the seeds for other factors, thoroughly eliminating the memory overhead associated with multivariate lifting.)

\noindent \textbf{Step 3: Scalar Lifting of the Single Seed (to $\mathbb{Z}_{13^2}$).}
Instead of lifting polynomials, we explicitly lift this single scalar seed $S_1$. First, we compute the invariant update multiplier $u_1 \equiv -[V'(S_1^{(1)})]^{-1} \pmod{13}$:
\[
u_1 \equiv -[3(3)^2 - 2(3) - 2]^{-1} \equiv -[19]^{-1} \equiv 2 \pmod{13}.
\]
We then compute the lifting error $\Delta_1 = V(S_1^{(1)}) / 13 = 13 / 13 = 1$ and perform the seed update:
\[
S_1^{(2)} = S_1^{(1)} + (u_1 \cdot \Delta_1 \bmod 13) \cdot 13 = 3 + 2 \times 13 = 29.
\]

\noindent \textbf{Step 4: Main Generator Recovery and Recurrent Expansion.}
We recover the main generator $A_1^{(2)}$ from $S_1^{(2)}$. Solving the equation $(A_1^{(2)})^2 \equiv 2 - S_1^{(2)} \equiv -27 \pmod{169}$ under the constraint $A_1^{(2)} \equiv A_1 \equiv 5 \pmod{13}$ yields $A_1^{(2)} = 135$.

Subsequently, using the Dickson recurrence $A_i^{(2)} = A_1^{(2)} A_{i-1}^{(2)} - A_{i-2}^{(2)} \pmod{169}$, we seamlessly stream the coefficients of all remaining factors:
\[
A_2^{(2)} \equiv 135^2 - 2 \equiv 140 \pmod{169}, \quad A_3^{(2)} \equiv 135 \times 140 - 135 \equiv 6 \pmod{169}.
\]
Thus, we obtain the complete factorization over $\mathbb{Z}_{13^2}$:
\[
x^{14}-1 = (x-1)(x+1)(x^2 \pm 135x + 1)(x^2 \pm 140x + 1)(x^2 \pm 6x + 1).
\]

\subsection{Example 2: Factorization of $x^{20}-1$ over $\mathbb{Z}_{19^3}$}

This example demonstrates a multi-step lift ($e=3$) for $p=19$.
Here $k=4$, and the auxiliary polynomial is $V(x) = x^4 - 3x^2 + 1$.

\noindent \textbf{Step 1: Single Seed Extraction.}
Using the main generator $A_1 = 6$, under the single-seed mechanism we only extract the corresponding scalar seed:
\[
S_1^{(1)} = 2 - A_1^2 \equiv 2 - 36 \equiv 4 \pmod{19}.
\]

\noindent \textbf{Step 2: Iterative Lifting of the Single Seed.}
We exclusively lift this single seed $S_1$ across all subsequent precision layers.
\begin{itemize}
    \item \textbf{Lift to $\mathbb{Z}_{19^2}$:}
    Using $u_1 \equiv 14$, we calculate the update for the first step: $S_1^{(2)} = 4 + 14 \cdot V(4) = 42$.
    The corresponding coefficient is $A_1^{(2)} = 120$.
    \item \textbf{Lift to $\mathbb{Z}_{19^3}$:}
    We evaluate the higher-order error of $V(x)$ at the new point $S_1^{(2)} = 42$ (with $\Delta_2 = V(42) / 19^2 = 8605 \equiv 17 \pmod{19}$):
    \[
    S_1^{(3)} = 42 + (u_1 \cdot \Delta_2 \bmod 19) \cdot 19^2 = 42 + (14 \times 17 \bmod 19) \cdot 361 = 3652.
    \]
    Recovering $A$ gives $A_1^{(3)} = 6618$.
\end{itemize}
The complete factorization over $\mathbb{Z}_{19^3}$ is then obtained by generating the remaining coefficients via the Dickson recurrence using $A_1^{(3)}$, totally avoiding any polynomial division or extended Euclidean algorithms.

\section{Performance Evaluation}
\label{sec:performance}

To empirically validate the efficiency of our \texttt{Dickson-Engine}, we conducted a comparative benchmark against \texttt{NTL}~\cite{ntl} (a standard number theory library) on an Intel Core Ultra 9 workstation.

\subsection{Complexity vs. Prime Size ($p$)}

We first evaluated the factorization time over the base field $\mathbb{F}_p$ ($e=1$) for primes ranging from 100 to 10,000. For \texttt{NTL}, we invoked the standard \texttt{CanZass} (Cantor-Zassenhaus) algorithm. For our \texttt{Dickson-Engine}, to ensure a rigorously fair ``cold-start'' comparison against \texttt{NTL}, we measured the total wall-clock time inclusive of the \textbf{random seed search} and the recursive generation of coefficients. By explicitly accounting for this \textbf{self-imposed initialization overhead}, we ensure that our performance metrics reflect a complete, standalone execution cycle rather than a pre-optimized ideal.

As illustrated in Figure~\ref{fig:perf_p}, \texttt{NTL} exhibits a characteristic super-linear growth in execution time (typically $O(p^2 \log p)$), exceeding 5 seconds when $p \approx 10,000$. In stark contrast, our \texttt{Dickson-Engine} maintains a near-constant execution time (consistently below 0.02s). This empirical evidence confirms the \textbf{$O(p)$ linear complexity} of our scalar recurrence approach. The achieved speedup, surpassing \textbf{300$\times$} at $p=10009$, validates that our structural generation successfully bypasses the heavy polynomial arithmetic bottlenecks inherent in generic solvers.

\subsection{Lifting Cost vs. Precision ($e$)}

A key contribution of our work is the efficient lifting to the ring $\mathbb{Z}_{p^e}$. Given that \texttt{NTL} is optimized for field arithmetic and lacks native support for polynomial factorization over rings (typically resulting in non-invertible element errors), we analyzed the scalability of our engine independently.

Figure~\ref{fig:perf_e} depicts the execution time for a fixed prime ($p = 1009$) as the lifting precision $e$ scales from 1 to 100. The results demonstrate a strictly \textbf{linear trend} ($O(e)$). 
Minor fluctuations (e.g., the dip at $e = 50$) are artifacts of the \textbf{initialization wrapper} used for benchmarking rather than the lifting logic itself; finding a valid seed earlier in the probabilistic search phase reduces the total overhead, whereas the structural lifting time remains invariant. This confirms that our lifting process, driven by the auxiliary polynomial $V(x)$, is computationally inexpensive and highly scalable.

\begin{figure}[ht!]
\centering
\subfloat[Time vs. Prime Characteristic ($p$). Comparison with NTL.]{%
    \label{fig:perf_p}%
    \includegraphics[width=0.48\textwidth]{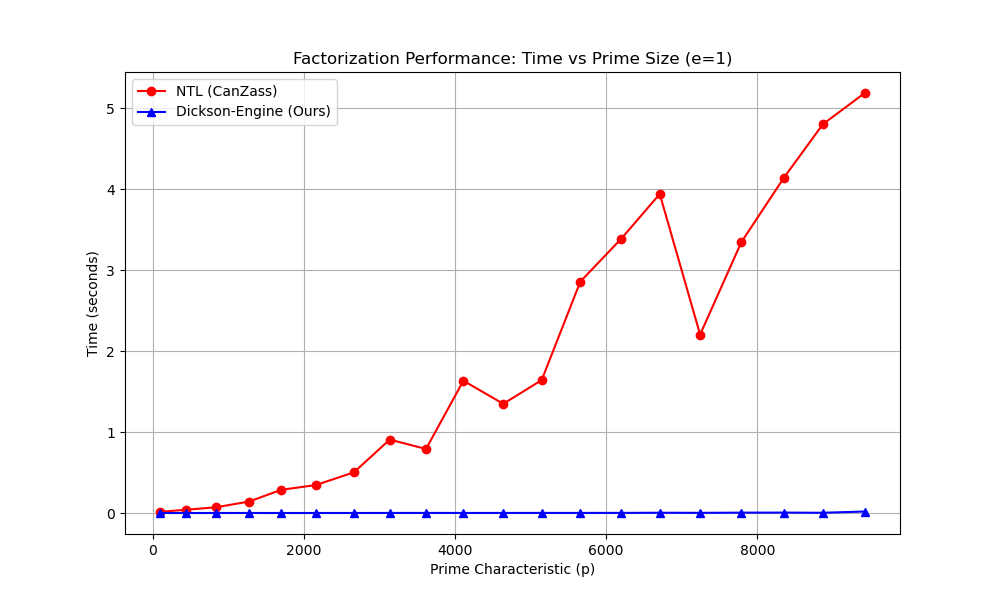}
}
\hfill 
\subfloat[Time vs. Precision ($e$). Scalability of Dickson-Engine.]{%
    \label{fig:perf_e}%
    \includegraphics[width=0.48\textwidth]{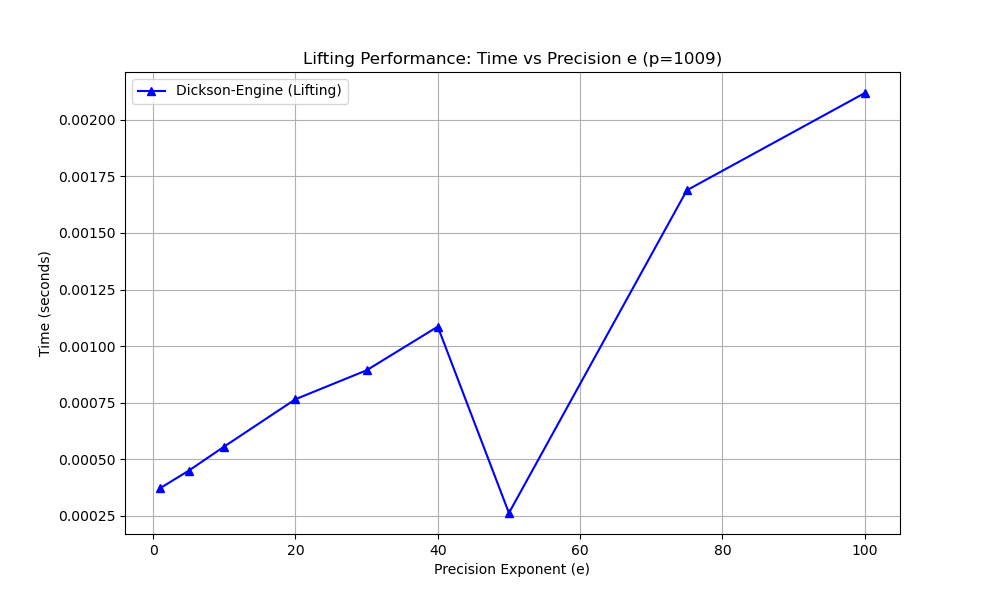}
}
\caption{Performance evaluation of the proposed algorithm. (a) Our method achieves orders-of-magnitude speedup over \texttt{NTL} for base field factorization. (b) The lifting process scales linearly with the ring precision $e$.}
\label{fig:performance_full}
\end{figure}

\section{Construction Results and Analysis}
\label{sec:construction}

\subsection{Setup and Methodology}
We applied our \texttt{Dickson-Engine} to the case of $p=13, e=2$, targeting the factorization of $x^{14}-1$ over the ring $\mathbb{Z}_{169}$. Using the proposed structural algorithm, we efficiently obtained all irreducible factors.
To evaluate the error-correcting capability, we mapped the constructed codes from the ring $\mathbb{Z}_{169}$ to the finite field $\mathbb{F}_{13}$ using the isometric \textbf{Gray map} defined in Section~\ref{sec:gray_map}. The resulting codes are linear codes over $\mathbb{F}_{13}$ with length $N = 14 \times 13 = 182$.

All computational experiments, including the performance benchmarks against NTL and the exhaustive code search, were conducted on a workstation equipped with an \textbf{Intel Core\texttrademark{} Ultra 9 285K} processor (24 cores) and \textbf{128 GiB} of RAM, running \textbf{Ubuntu 20.04.6 LTS}. The multi-core search was implemented in Python using the \texttt{multiprocessing} library.

\subsection{Optimal Parameters for Classical LCD Codes}

\subsubsection*{Structure and Linearity of the Constructed Codes}

To prevent ambiguity regarding the algebraic nature of the constructed codes, we clarify their structure as follows:
\begin{itemize}
    \item \textbf{Algebraic Domain (The Ring):} The codes are constructed as cyclic codes of length $n=14$ over the integer residue ring $\mathbb{Z}_{13^2}$. Algebraically, each code $\mathcal{C}$ corresponds to an ideal $\langle g(x) \rangle$ in the quotient ring $\mathbb{Z}_{13^2}[x]/\langle x^{14}-1 \rangle$. Since we select generators satisfying the LCD condition, these codes are free modules over $\mathbb{Z}_{13^2}$.
    
    \item \textbf{Physical Domain (The Field):} The parameters reported (e.g., length $N=182$, minimum distance $d$) refer to the image of $\mathcal{C}$ under the isometric Gray map $\Phi: \mathbb{Z}_{13^2}^{14} \to \mathbb{F}_{13}^{182}$. While the pre-image $\mathcal{C}$ is linear over the ring, the resulting code $\Phi(\mathcal{C})$ in the ambient space $\mathbb{F}_{13}$ is a \textbf{$\mathbb{Z}_{13^2}$-linear code} (a specific class of potentially non-linear codes in the vector space sense) that inherits the distance properties of the ring structure.
\end{itemize}

Our primary focus is the construction of \textbf{Linear Complementary Dual (LCD) codes}. According to the criteria established by Liu and Liu~\cite{Liu2015} (Corollary 4.8), a cyclic code over a finite chain ring is LCD if and only if its generator polynomial is self-reciprocal (assuming coprime length). Since all irreducible factors of $x^{14}-1$ over $\mathbb{Z}_{169}$ derived via our Dickson method are inherently \textbf{self-reciprocal}, any code generated by their product is guaranteed to be an \textbf{LCD code}.

We performed an exhaustive search (Brute Force) for small dimensions ($k \le 4$) and a massive parallel random sampling ($1.2 \times 10^8$ samples per code dimension) for higher dimensions ($k > 4$). The results are summarized in Figure~\ref{fig:lcd_performance}.

\begin{figure}[htbp]
    \centering
    \includegraphics[width=0.8\textwidth]{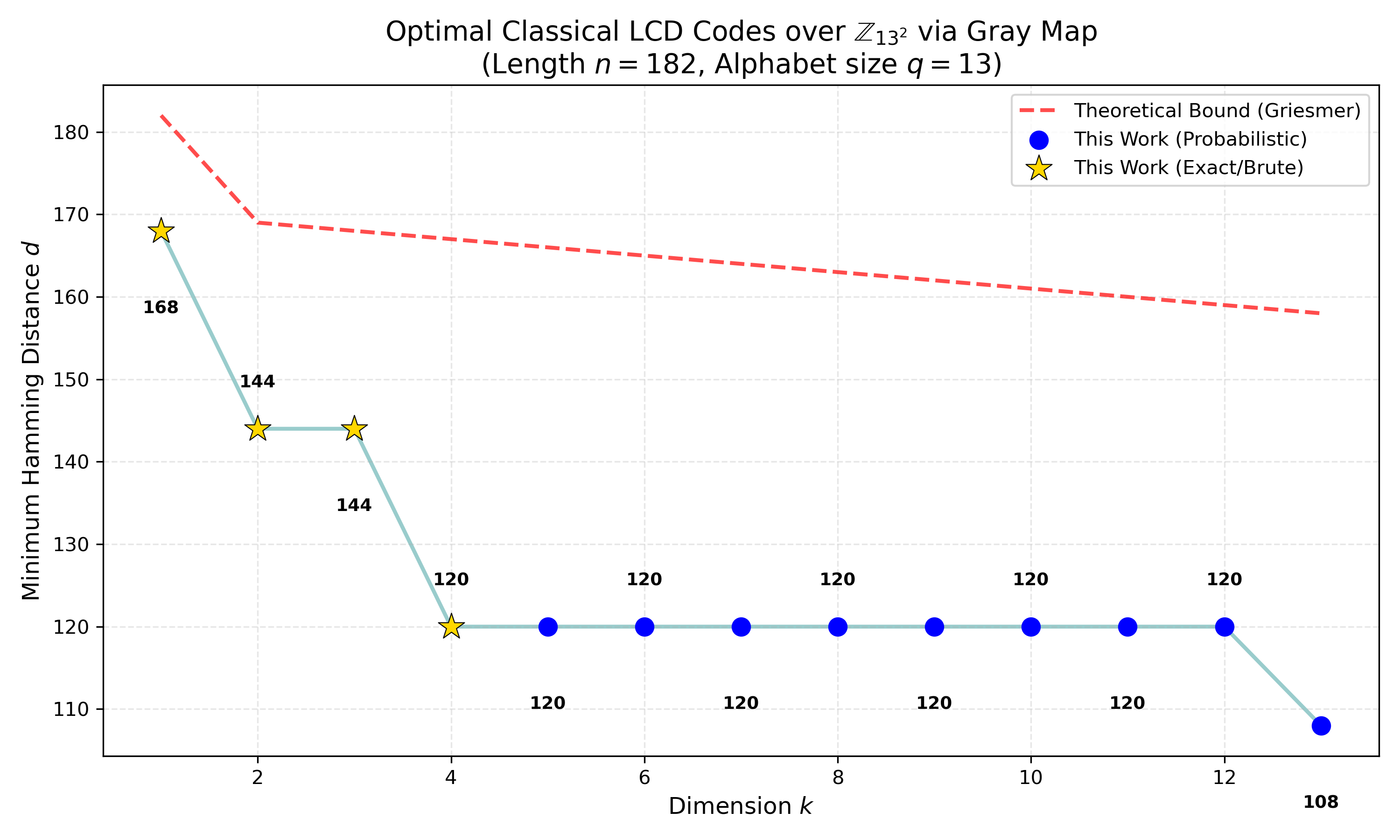}
    \caption{Performance profile of constructed LCD codes over $\mathbb{F}_{13}$ with length $N=182$. The plot contrasts the achieved minimum distance against the theoretical Griesmer Bound.}
    \label{fig:lcd_performance}
\end{figure}

\noindent \textbf{Key Observations:}
\begin{enumerate}
    \item \textbf{Near-Optimal Performance ($k \le 3$):}
    For small dimensions, our constructed codes achieve minimum distances $d$ that are remarkably close to the theoretical \textbf{Griesmer Bound}. Specifically, the codes $[182, 1, 168]_{13}$ and $[182, 2, 144]_{13}$ are confirmed by exhaustive search to be \textbf{optimal or near-optimal} within the cyclic algebraic structure.

    \item \textbf{The ``Robustness Plateau'' ($k=4 \to 12$):}
    A striking feature of these codes is the persistence of high minimum distance ($d \approx 120$) as the dimension $k$ increases from 4 to 12. Although the values for $k>4$ serve as probabilistic lower bounds on the true distance, the stability of the observed weight spectrum reveals an \textbf{intrinsic structural property} of cyclic codes over local rings that, to our knowledge, has not been previously documented. This phenomenon suggests that the algebraic structure of codes over $\mathbb{Z}_{p^2}$ still holds rich, unexplored phenomena---and our factorization engine provides precisely the \textbf{computational foundation} needed to investigate them at scale.

    \item \textbf{Symmetry Breaking:}
    To make this observation concrete, we list the complete factorization over $\mathbb{Z}_{169}$ and identify the conjugate pairing structure. Let $f_i$ denote the $i$-th irreducible factor (indexed $0$--$7$):

    \begin{table}[h]
    \centering
    \caption{Irreducible factors of $x^{14}-1$ over $\mathbb{Z}_{169}$ and their conjugate pairing.}
    \label{tab:factors}
    \renewcommand{\arraystretch}{1.2}
    \begin{tabular}{c|l|c}
    \hline
    \textbf{Index} & \textbf{Factor $f_i$} & \textbf{Conjugate Pair} \\
    \hline
    $0$ & $x - 1$ & --- \\
    $1$ & $x + 1$ & --- \\
    $2$ & $x^2 + 34x + 1$ & \multirow{2}{*}{$\{f_2, f_3\}$} \\
    $3$ & $x^2 + 135x + 1$ &  \\
    $4$ & $x^2 + 29x + 1$ & \multirow{2}{*}{$\{f_4, f_5\}$} \\
    $5$ & $x^2 + 140x + 1$ &  \\
    $6$ & $x^2 + 163x + 1$ & \multirow{2}{*}{$\{f_6, f_7\}$} \\
    $7$ & $x^2 + 6x + 1$ &  \\
    \hline
    \end{tabular}
    \end{table}
    
    Note that the coefficients within each conjugate pair are negatives of each other modulo 169 (e.g., $34 + 135 = 169$ and $29 + 140 = 169$). Selecting a generator polynomial $g(x)$ that preserves complete conjugate pairs forces the middle terms of each quartic product $f_i f_j = x^4 + S_i x^2 + 1$ to cancel, yielding a \textbf{sparse} generator.

    \textbf{Exhaustive comparison at $k=4$.} Our brute-force search over all $\binom{8}{6}=28$ choices of 6-factor generator polynomials reveals a stark dichotomy:
    \begin{itemize}
        \item \textbf{Symmetric selections} (all conjugate pairs intact), e.g., $g(x) = f_0 f_1 f_2 f_3 f_4 f_5$: yield $d = 72$.
        \item \textbf{Broken selections} (at least one pair split), e.g., $g(x) = f_0 f_1 f_2 f_3 f_4 f_6$: yield $d = 120$, a \textbf{66\% improvement}.
    \end{itemize}
    This pattern is remarkably consistent: among all 28 configurations at $k=4$, every selection that keeps all conjugate pairs intact produces $d \le 84$, while every selection that breaks at least one pair achieves $d \ge 96$.
\end{enumerate}

\subsubsection*{Design Insight: Sparsity vs.\ Density of Generator Polynomials}

The mechanism behind this phenomenon admits an intuitive explanation rooted in the structure of the generator polynomial $g(x)$.
When conjugate pairs $\{f_i, f_j\}$ with $A_j = -A_i$ are selected together, their product $(x^2 - A_i x + 1)(x^2 + A_i x + 1) = x^4 + S_i x^2 + 1$ cancels the odd-degree terms. The resulting generator $g(x)$ is a \textbf{sparse polynomial} with many zero coefficients. In the codeword space, this sparsity implies the existence of low-weight codewords---the ``skeleton'' of $g(x)$ can be recreated by a small number of non-zero positions, resulting in a low minimum distance.

Conversely, breaking a conjugate pair forces the retention of odd-degree terms in the factor product, producing a \textbf{dense polynomial} with non-zero coefficients spread across all degrees. This density acts as algebraic \textbf{redundancy}: the information is ``smeared'' uniformly across every coordinate of the codeword, creating a robust error-correcting structure that is difficult to penetrate with low-weight error patterns.

In essence, \textit{algebraic symmetry yields computational elegance, but breaking this symmetry provides the dense redundancy required for optimal error correction.} This principle may serve as a useful heuristic for future code constructions over local rings.

\subsection{The Advantage of Native Cyclic Structure}
Since the constructed LCD codes have trivial hull dimension ($c=0$), they naturally serve as candidates for standard quantum codes without entanglement consumption.

It is worth noting that while recent works have elegantly demonstrated that linear codes over finite fields can be transformed into LCD codes~\cite{carlet2018linear} or EAQECCs with desired entanglement parameters~\cite{anderson2024relative} via monomial equivalence, such column-scaling transformations generally destroy the highly desirable \textbf{cyclic structure}. In contrast, our Dickson-based factorization natively produces self-reciprocal generator polynomials. This allows for the direct construction of structurally pristine cyclic LCD codes (and pure quantum codes) over the local ring $\mathbb{Z}_{p^e}$, preserving the shift-invariance that is critical for efficient hardware implementation without any structural degradation.

\section{Conclusion and Future Work}

In this paper, we have established a novel structural connection between Dickson polynomials and the explicit factorization of $x^{p+1}-1$ over $\mathbb{Z}_{p^e}$. 
The resulting algorithm, implemented as the open-source \texttt{Dickson-Engine}, demonstrates a significant computational advantage over standard methods and enabled the construction of a family of near-optimal classical LCD codes. Our analysis of the factor indices further revealed that breaking the symmetry of conjugate pairs is a critical design principle for maximizing the minimum distance.

Our work points towards two primary directions for future investigation. 
The most fundamental challenge is to generalize the theoretical framework presented here. 
Our discovery of the $V(x)$ mechanism for the specific case of $n=p+1$ raises a crucial question: \textbf{can this ``structural lifting'' approach be extended to other classes of polynomials?}
A particularly important target is the family of general \textbf{Cyclotomic Polynomials, $\Phi_n(x)$}. 
Investigating whether a similar recursive structure governs the lifting coefficients of $\Phi_n(x)$ over residue rings could lead to a unified theory for the explicit construction of a much broader class of cyclic codes.

From a practical standpoint, further optimization within the current framework is also warranted. 
Future work includes extending this ``structural lifting'' framework to general cyclotomic polynomials $\Phi_n(x)$ and exploring alternative \textbf{Gray isometries} to further optimize non-linear properties in higher-order rings ($\mathbb{Z}_{p^e}$ for $e \ge 3$). 
A rigorous analysis of the \textbf{non-linear properties} of the resulting code families in $\mathbb{F}_p$ may also yield parameters that provably surpass the established bounds for linear codes under specific conditions.

\section*{Acknowledgment}
This work was supported by the NSFC (Nos. 12271084, U2241216) and the Key Laboratory of Safety-Critical Software Development and Verification, MIIT. 
We deeply appreciate Google Gemini~\cite{gemini2026} for its invaluable assistance in linguistic refinement and code optimization. All core algebraic insights, mathematical validations, and the final integrity of this work remain the sole responsibility of the human authors.

\bibliographystyle{plain}
\bibliography{ref}

\appendix
\begin{appendix}
\section{Mathematical Proof of the Structural Isomorphism}
\label{app:proof}

In this appendix, we provide the rigorous proof for Theorem 3.2. The core logic relies on establishing that the polynomial $V(x)$, defined via specific binomial coefficients, implicitly satisfies the recurrence relations governing the roots of the factorization.

\subsection{Combinatorial Foundation}

We first introduce a necessary combinatorial lemma, which relates binomial coefficients in a summation identity.

\begin{lem} \label{lemma:comb}
Let $j, k$ be positive integers with $k \ge 2j$. Then the following identity holds:
\begin{equation}
    \sum_{i=0}^{j} (-1)^i \binom{k-2i}{j-i} \binom{k-i}{i} = 1.
\end{equation}
\end{lem}
\begin{proof}
This is a specific instance of identities involving alternating sums of binomial coefficients. It can be derived by considering the coefficient extraction from the generating function of Chebyshev polynomials or by induction. For brevity, we refer readers to standard combinatorial texts (e.g., [13]) for the detailed algebraic manipulation.
\end{proof}

\subsection{Proof of Theorem 3.2}

\textbf{Goal:} We aim to show that the roots of $V(x) = \sum_{r=0}^k V_r x^{k-r}$ are exactly the structural values $S_i^{(h)} = 2 - (A_i^{(h)})^2$ in $\mathbb{Z}_{p^h}$.

\textbf{Step 1: Establishing the System of Equations.}
Recall that the factorization of $x^{p+1}-1$ (specifically the part corresponding to size-2 cosets) can be written in terms of $S_i$. For the case $p = 4k+3$, we have:
\begin{equation}
    \prod_{i=1}^k (x^4 + S_i^{(h)} x^2 + 1) \equiv \sum_{j=0}^k x^{4j} \pmod{p^h}.
    \label{eq:factor_expansion}
\end{equation}
Let $W_r$ denote the signed elementary symmetric polynomials of the roots $\{S_1^{(h)}, \dots, S_k^{(h)}\}$:
\begin{equation}
    W_r = (-1)^r \sum_{1 \le i_1 < \dots < i_r \le k} S_{i_1}^{(h)} \dots S_{i_r}^{(h)}.
\end{equation}
Expanding the left side of Eq.~(\ref{eq:factor_expansion}) and collecting terms by powers of $x$, the coefficient of $x^{4j}$ (where $0 \le j \le k$) must match the right side (which is $1$). By comparing coefficients, we obtain a linear system governing $W_r$. Specifically, for the terms involving even powers, we derive:
\begin{equation}
    \sum_{i=0}^j \binom{k-2i}{j-i} W_{2i} \equiv 1 \pmod{p^h}, \quad \text{for } 0 \le j \le \lfloor k/2 \rfloor.
    \label{eq:linear_system}
\end{equation}

\textbf{Step 2: Verification of $V(x)$.}
We need to prove that the coefficients $V_r$ defined in Definition 2 satisfy the system in Eq.~(\ref{eq:linear_system}). In other words, we need to show that $W_{2i} \equiv V_{2i}$.
Recall from Definition 2 (for $p=4k+3$) that:
\begin{equation}
    V_{2i} = (-1)^i \binom{k-i}{i}.
\end{equation}
Substituting $V_{2i}$ into the left side of Eq.~(\ref{eq:linear_system}):
\begin{equation}
    \sum_{i=0}^j \binom{k-2i}{j-i} \left[ (-1)^i \binom{k-i}{i} \right] = \sum_{i=0}^j (-1)^i \binom{k-2i}{j-i} \binom{k-i}{i}.
\end{equation}
By invoking \textbf{Lemma \ref{lemma:comb}}, this summation identically equals $1$.
Thus, the coefficients $V_r$ are indeed the elementary symmetric polynomials of the roots $S_i^{(h)}$. By Vieta's formulas, this implies that $S_i^{(h)}$ are the roots of the polynomial $V(x)$.

The proof for the case $p=4k+1$ follows an analogous derivation with slightly different binomial identities. This confirms that the lifting condition $V(S_i^{(h)}) \equiv 0 \pmod{p^h}$ is structurally isomorphic to preserving the factorization. \qed

\end{appendix}

\end{document}